\documentclass[11pt,letterpaper]{article}
\usepackage{amsmath}
\usepackage{amssymb}
\usepackage{fullpage}
\usepackage{color}
\usepackage{xspace}
\usepackage{hyperref}
\usepackage[amsthm,amsmath,thmmarks]{ntheorem}

\definecolor{hyperblue}{RGB}{0,0,88}
\hypersetup{pdfborder={0 0 0},colorlinks=true,linkcolor=hyperblue,citecolor=hyperblue,filecolor=hyperblue,urlcolor=hyperblue}

\newcommand{\refmainone}{\hyperref[nnr:thm:main]{Theorem~\ref*{nnr:thm:main}(1)}\xspace}
\newcommand{\refmaintwo}{\hyperref[nnr:thm:main]{Theorem~\ref*{nnr:thm:main}(2)}\xspace}
\newcommand{\refmainthree}{\hyperref[nnr:thm:main]{Theorem~\ref*{nnr:thm:main}(3)}\xspace}

\newtheorem{theorem}{Theorem}

\newtheorem{lemma}{Lemma}

\newtheorem{observation}{Observation}
\newtheorem{conjecture}{Conjecture}
\newtheorem{definition}{Definition}

\newcommand{\UP}{{\rm UP}}
\newcommand{\NP}{{\rm NP}}
\newcommand{\LWPP}{{\rm LWPP}}
\newcommand{\parityP}{\oplus{\rm P}}

\newcommand{\rank}{\operatorname{rank}}
\newcommand{\nnr}{\operatorname{rank}_+}
\newcommand{\binr}{\operatorname{rank}_{0,1}}
\newcommand{\fool}{\operatorname{fool}}

\begin{document}

\title{Nonnegative Rank vs.~Binary Rank}
\author{Thomas Watson\footnote{Department of Computer Science, University of Toronto. Supported by funding from NSERC.}}
\maketitle

\begin{abstract}
Motivated by (and using tools from) communication complexity, we investigate the relationship between the following two ranks of a $0$-$1$ matrix: its nonnegative rank and its binary rank (the $\log$ of the latter being the unambiguous nondeterministic communication complexity). We prove that for partial $0$-$1$ matrices, there can be an exponential separation. For total $0$-$1$ matrices, we show that if the nonnegative rank is at most $3$ then the two ranks are equal, and we show a separation by exhibiting a matrix with nonnegative rank $4$ and binary rank $5$, as well as a family of matrices for which the binary rank is $4/3$ times the nonnegative rank.
\end{abstract}


\section{Introduction}
\label{nnr:sec:intro}

Let $M$ be a total $0$-$1$ matrix of size $n\times m$. We consider three notions of the rank of $M$ (over the reals).
\begin{itemize}
\item[(i)] The familiar \emph{real rank}, denoted $\rank(M)$, is the smallest $r$ for which there exists a \emph{product} decomposition $M=UV$ where $U$ is $n\times r$ and $V$ is $r\times m$. Equivalently, it is the smallest $r$ for which there exists a \emph{sum} decomposition $M=M^{(1)}+\cdots+M^{(r)}$ where each component is rank $1$, say $M^{(i)}=u^{(i)}v^{(i)}$ where $u^{(i)}$ is $n\times 1$ and $v^{(i)}$ is $1\times m$.
\item[(ii)] The \emph{nonnegative rank}, denoted $\nnr(M)$, is defined in the same way but with the restriction that $U$ and $V$ are entry-wise nonnegative, or equivalently that all $u^{(i)},v^{(i)}$ (and hence $M^{(i)}$) are entry-wise nonnegative.
\item[(iii)] The \emph{binary rank}, denoted $\binr(M)$, is defined in the same way but with the further restriction that all entries of $U$ and $V$ come from $\{0,1\}$, or equivalently that all entries of $u^{(i)},v^{(i)}$ (and hence $M^{(i)}$) come from $\{0,1\}$.
\end{itemize}
We define a \emph{partial} $0$-$1$ matrix as having entries from $\{0,1,*\}$, where $*$ entries are wildcards representing arbitrary numbers (from some set depending on the context; see details below). The real, nonnegative, or binary rank of a partial $0$-$1$ matrix $M$ are defined (respectively) as the minimum real, nonnegative, or binary rank of any total real matrix that agrees with $M$ on the non-$*$ entries. In the case of nonnegative rank, the total matrix must have nonnegative entries, and in the case of binary rank, the total matrix must have nonnegative integer entries (but in any case, its entries corresponding to $*$ need not come from $\{0,1\}$). Note that for all $0$-$1$ matrices $M$ (total or partial), $\rank(M)\le\nnr(M)\le\binr(M)$.

These ranks are closely related to measures of communication complexity of the two-party function associated with $M$. (See \cite{KN,Juk} for background on communication complexity.) The binary rank is just the number of (combinatorial) rectangles needed to partition the $1$'s of $M$. This corresponds to unambiguous nondeterminism (the communication complexity analogue of the classical complexity class $\UP$), so the unambiguous nondeterministic communication complexity\footnote{This is closely related to the clique vs.~independent set family of problems, since it is well-known that every two-party total function reduces to the clique vs.~independent set problem for a graph on $\binr(M)$ nodes.} is defined to be $\log_2\binr(M)$.\footnote{There are two other types of rank defined in the same way as binary rank but with different arithmetic. With boolean arithmetic (so $+$ is OR) this is sometimes known as \emph{boolean rank}; it is the number of rectangles needed to cover the $1$'s of $M$; it lower bounds the nonnegative rank and corresponds to $\NP$. With $GF(2)$ arithmetic (so $+$ is XOR) this is sometimes known as \emph{xor rank}; it is the number of rectangles needed to cover each $1$ of $M$ an odd number of times and each $0$ of $M$ an even number of times; it lower bounds the real rank and corresponds to $\parityP$.}

Similarly, $\log_2\rank(M)$ can be viewed as a communication measure corresponding to the classical complexity class $\LWPP$ (representing randomized algorithms whose acceptance probability is one common value on all $0$-inputs and a different common value on all $1$-inputs). The famous \emph{Log-Rank Conjecture} asserts that $\log\rank(M)$ is polynomially related to the deterministic communication complexity of $M$, for every total $0$-$1$ matrix $M$.

Analogously to the correspondences between $\log_2\binr$ and $\UP$, and between $\log_2\rank$ and $\LWPP$, we can view $\log_2\nnr$ as corresponding to a classical complexity class representing randomized algorithms whose acceptance probability is $0$ on all $0$-inputs and a positive common value on all $1$-inputs (though this class appears to lack a name). Nonnegative rank of a total $0$-$1$ matrix $M$ also corresponds to the following random sampling problem: Alice and Bob are given private randomness but no input, and their goal is for Alice to output a row index and Bob to output a column index such that the corresponding entry of $M$ is uniformly distributed over all the $1$'s of $M$.\footnote{A closely related type of problem (which was studied in the communication complexity setting in \cite{ASTVW}, and in the setting of constant-depth circuits in \cite{Vio}) is to sample a uniformly random input-output pair of a given function. In our problem we just ask to sample a uniformly random $1$-input, which can be viewed as peripherally motivated by the topic of uniform sampling / approximate counting of $\NP$ witnesses.} It turns out that one-way communication is optimal in this setting, and the nonnegative rank is exactly the minimum number of transcripts needed to solve this problem (see, e.g., \cite{JSWZ}, where it is called ``correlation complexity''), and hence the minimum number of bits of communication is $\log_2\nnr(M)$. If $M$ is a partial $0$-$1$ matrix, then the associated sampling problem is for Alice's row and Bob's column to be uniformly distributed over the $1$ entries of $M$, \emph{conditioned} on being a non-$*$ entry.

We consider the relationship between nonnegative rank and binary rank (hence between the above sampling problem and unambiguous nondeterminism). In terms of motivation, nonnegative rank of real matrices has myriad applications in theoretical computer science, other branches of computer science, and other scientific disciplines (see \cite{CR,Gil}), but there seems to be a dearth of proof techniques that directly exploit special properties of nonnegative rank. For example, there seem to be almost no known general-purpose techniques for upper bounding $\nnr$ without upper bounding $\binr$, or for lower bounding $\binr$ without lower bounding $\nnr$. We provide some (albeit ad hoc) techniques that differentiate between $\nnr$ decompositions and $\binr$ decompositions. Finally, we feel it is always natural to compare different measures of complexity, and from a purely combinatorial point of view it is natural to compare different notions of rank of $0$-$1$ matrices.

It is known that for every total $0$-$1$ matrix $M$,
\begin{align}
\log_2\binr(M)~&\le~\text{deterministic communication complexity of $M$}\nonumber\\
&\le~O\bigl((\log\rank(M))\cdot(\text{nondeterministic communication complexity of $M$})\bigr)\nonumber\\
&\le~O(\log^2\nnr(M))\nonumber\\
&\le~O(\log^2\binr(M))\label{nnr:eq:quadratic}
\end{align}
where the second line follows by a result of \cite{LS}. (A simple proof of the special case ``deterministic communication complexity of $M\le O(\log^2\binr(M))$'' was given earlier in \cite{Yan}.) We show that the quadratic upper bound $\log\binr(M)\le O(\log^2\nnr(M))$ does not hold in general for partial $0$-$1$ matrices; there is an exponential separation. Our proof of this exploits machinery from \cite{GLMWZ,Goos} in a new way and also (perhaps surprisingly) exploits the upper bound \eqref{nnr:eq:quadratic} for total $0$-$1$ matrices. Prior to this work it was open whether nonnegative rank and binary rank are \emph{always equal} for total $0$-$1$ matrices; we refute this, although our counterexamples are very far from resolving whether the quadratic upper bound is tight.

\begin{theorem}
\label{nnr:thm:partialsep}
There exists a family of partial $0$-$1$ matrices $M$ such that $\binr(M)\ge 2^{\nnr(M)^{\Omega(1)}}$.
\end{theorem}

\begin{theorem}
\label{nnr:thm:main}
~\begin{itemize}
\item[(1)] For every total $0$-$1$ matrix $M$, if $\nnr(M)\le 3$ then $\nnr(M)=\binr(M)$.
\item[(2)] There exists a total $0$-$1$ matrix $M$ such that $\nnr(M)=4$ and $\binr(M)=5$.
\item[(3)] There exists a total $0$-$1$ matrix $M$ such that $\nnr(M)=9$ and $\binr(M)=12$.
\end{itemize}
\end{theorem}

Tensoring the matrix from \refmainthree with the identity shows that for every $k$ divisible by $9$ there exists a total $0$-$1$ matrix with nonnegative rank $k$ and binary rank $(4/3)k$. After taking logs, this yields an extremely meager gap between the sampling and unambiguous nondeterminism complexity measures (an additive constant less than $1$, which would sometimes be wiped out by taking ceilings, anyway). The primary open question is whether a better gap can be exhibited. A natural approach is to amplify the gap by repeatedly tensoring a matrix with itself. Lower bounding the binary rank of the resulting matrix seems related to the notoriously difficult direct sum problem for deterministic communication complexity (since the $k$-fold tensor product corresponds to taking the AND of the outputs of $k$ independent instances of the original two-party function, which is no harder than computing the values of all $k$ outputs). We have been unable to prove any interesting such lower bounds, even for the concrete matrices from \refmaintwo and \refmainthree. This is further discussed in \autoref{nnr:sec:open}.

There are several other works investigating the relationships of different ranks with each other and with measures of communication complexity (for total matrices). G\"o\"os, Pitassi, and Watson \cite{GPW} exhibited a nearly quadratic (essentially tight) gap between deterministic communication complexity and the log of the binary rank, as well as a nearly power $1.5$ gap between deterministic communication complexity and the log of the number of monochromatic rectangles needed to partition the $0$'s and $1$'s of the matrix (i.e., the log of the sum of the binary rank and the binary rank of the complement matrix). That power $1.5$ gap (which has subsequently been improved to quadratic by Ambainis, Kokainis, and Kothari \cite{AKK}) improves the factor $2$ gap due to Kushilevitz, Linial, and Ostrovsky \cite{KLO}, which was also the previous record for deterministic communication complexity vs.~log of the binary rank (and we mention that for the latter, Kushilevitz and Weinreb \cite{KW} had used different techniques to obtain a weaker factor $1.12$ gap). For deterministic communication complexity vs.~log of the real rank, the best gap before \cite{GPW} was a power $\log_36\approx 1.63$ gap due to Nisan, Wigderson, and Kushilevitz \cite{NW}; in the other direction, Lovett \cite{Lov} showed that the deterministic communication complexity of any $M$ is $O\bigl(\sqrt{\rank(M)}\log\rank(M)\bigr)$. G\"o\"os \cite{Goos} exhibited matrices for which the conondeterministic communication complexity is at least a power $1.128$ greater than the log of the binary rank, improving on the factor $2$ gap due to Shigeta and Amano \cite{SA}. Shitov \cite{Shi} showed that for every $n\times m$ real nonnegative matrix with real rank $3$, the nonnegative rank is at most $\bigl\lceil\frac{6}{7}\cdot\min(n,m)\bigr\rceil$.


\section{Partial Matrices}
\label{nnr:sec:partial}

In this section we prove \autoref{nnr:thm:partialsep}. Letting $g\colon\{0,1\}^b\times\{0,1\}^b\to\{0,1\}$ be a total function (usually called a \emph{gadget}), define $g^n\colon\{0,1\}^{bn}\times\{0,1\}^{bn}\to\{0,1\}^n$ by $g^n(x,y)=\bigl(g(x_1,y_1),\ldots,g(x_n,y_n)\bigr)$ where $x=(x_1,\ldots,x_n)$ and $y=(y_1,\ldots,y_n)$ with $x_i,y_i\in\{0,1\}^b$ for each $i$. For any (possibly partial) function $f\colon\{0,1\}^n\to\{0,1\}$, the (possibly partial) two-party communication problem $f\circ g^n\colon\{0,1\}^{bn}\times\{0,1\}^{bn}\to\{0,1\}$ is defined by $(f\circ g^n)(x,y)=f\bigl(g(x_1,y_1),\ldots,g(x_n,y_n)\bigr)$, and we identify this with a $2^{bn}\times 2^{bn}$ partial $0$-$1$ matrix with rows indexed by $x$ and columns indexed by $y$.

Consider the partial function $f\colon\{0,1\}^n\to\{0,1\}$ (where $n$ is even) defined by
\begin{equation}
\label{nnr:eq:f}
f(z)~=~\begin{cases}1&\text{if }|z|=n/2\\0&\text{if }|z|=0\end{cases}
\end{equation}
where $|z|$ denotes the Hamming weight of $z$.

\begin{lemma}
\label{nnr:lem:upper}
For $f$ as in \eqref{nnr:eq:f} and for all $g\colon\{0,1\}^b\times\{0,1\}^b\to\{0,1\}$, we have $\nnr(f\circ g^n)\le n\cdot 2^b$.
\end{lemma}

\begin{proof}
For $i\in\{1,\ldots,n\}$ and $s\in\{0,1\}^b$, define the total function (matrix) $M^{(i,s)}\colon\{0,1\}^{bn}\times\{0,1\}^{bn}\to\mathbb{R}_{\ge 0}$ by \[M^{(i,s)}(x,y)~=~\begin{cases}2/n&\text{if $x_i=s$ and $g(x_i,y_i)=1$}\\0&\text{otherwise}\end{cases}.\] Since each $M^{(i,s)}$ is nonnegative and has rank $1$, it suffices to check that $\sum_{i,s}M^{(i,s)}$ agrees with $f\circ g^n$ on the latter function's domain. If $(x,y)$ is such that $|g^n(x,y)|=0$ then $M^{(i,s)}(x,y)=0$ for all $(i,s)$. If $(x,y)$ is such that $|g^n(x,y)|=n/2$ then $M^{(i,s)}(x,y)=2/n$ for exactly $n/2$ many $(i,s)$; specifically, there are $n/2$ many $i$'s for which $g(x_i,y_i)=1$, and for each such $i$ there is exactly one $s$ for which $x_i=s$.
\end{proof}

The paper \cite{GLMWZ} introduced and studied what we call the \emph{confounding} gadget $g\colon\{0,1\}^b\times\{0,1\}^b\to\{0,1\}$, defined by
\begin{equation}
\label{nnr:eq:g}
g(x_i,y_i)~=~\langle x_i,y_i\rangle\bmod{2}\qquad\text{and}\qquad b~=~b(n)~=~100\log_2n.
\end{equation}

\begin{lemma}
\label{nnr:lem:lower}
For $f$ as in \eqref{nnr:eq:f} and $g$ as in \eqref{nnr:eq:g}, we have $\binr(f\circ g^n)\ge 2^{\Omega(\sqrt{n\log n})}$.
\end{lemma}

\autoref{nnr:thm:partialsep} follows immediately from \autoref{nnr:lem:upper} (which shows that $\nnr(f\circ g^n)\le n^{101}$ for the confounding gadget $g$) and \autoref{nnr:lem:lower}. We conjecture that when $g$ is the AND gadget (with $b=1$), $\binr(f\circ g^n)\ge 2^{\Omega(n)}$, but we are unable to prove this.

To prove \autoref{nnr:lem:lower}, we use the following property of the confounding gadget, which is a special case of the main technical result in \cite{GLMWZ} (and a streamlined, self-contained proof of this special case appears in \cite{Goos}). A \emph{subcube} of codimension $d$ is defined as the subset of $\{0,1\}^n$ consisting of all $2^{n-d}$ strings consistent with some partial assignment that fixes $d$ of the bit positions. For $z\in\{0,1\}^n$, we also define the set $W_z\subseteq\{0,1\}^{bn}\times\{0,1\}^{bn}$ to be $(g^n)^{-1}(z)$.

\begin{lemma}
\label{nnr:lem:conf}
For $g$ as in \eqref{nnr:eq:g}, for every rectangle $X\times Y\subseteq\{0,1\}^{bn}\times\{0,1\}^{bn}$ and every $z\in\{0,1\}^n$, if $|(X\times Y)\cap W_z|\ge 2^{-c}\cdot|W_z|$, then there exists a subcube $Z\subseteq\{0,1\}^n$ of codimension $O(c/b)=O(c/\log n)$ such that $Z\subseteq g^n(X,Y)$ (i.e., for every $z'\in Z$ there exists an $(x,y)\in X\times Y$ such that $g^n(x,y)=z'$).
\end{lemma}

\begin{proof}[Proof of \autoref{nnr:lem:lower}]
Suppose for contradiction there is a collection of rectangles $R_i=X_i\times Y_i$ for $i\in\{1,\ldots,2^{o(\sqrt{n\log n})}\}$ such that if $(f\circ g^n)(x,y)=1$ then $(x,y)\in R_i$ for exactly one $i$, and if $(f\circ g^n)(x,y)=0$ then $(x,y)\in R_i$ for no $i$'s.

First we claim that for all $i<j$, $|R_i\cap R_j|\le 2^{-\Omega(n\log n)}\cdot 2^{2bn}$. Supposing not, we would have $|R_i\cap R_j\cap W_z|\ge 2^{-o(n\log n)}\cdot|W_z|$ for some $z$ (since the sets $W_z$ partition $\{0,1\}^{bn}\times\{0,1\}^{bn}$). Then by \autoref{nnr:lem:conf} there would exist a subcube $Z\subseteq g^n(R_i\cap R_j)$ of codimension $o((n\log n)/\log n)=o(n)\le n/2$. But since any subcube of codimension $\le n/2$ contains a $1$-input of our $f$, this means $R_i\cap R_j$ contains a $1$-input of $f\circ g^n$, contradicting our assumption. This proves the claim.

In particular, for all $i<j$ we have either $|X_i\cap X_j|\le 2^{-\Omega(n\log n)}\cdot 2^{bn}$ or $|Y_i\cap Y_j|\le 2^{-\Omega(n\log n)}\cdot 2^{bn}$; define $Q_{i,j}=(X_i\cap X_j)\times\{0,1\}^{bn}$ in the former case and $Q_{i,j}=\{0,1\}^{bn}\times(Y_i\cap Y_j)$ in the latter case. Then we have $|Q_{i,j}|\le 2^{-\Omega(n\log n)}\cdot 2^{2bn}$ and hence $\bigl|\bigcup_{i<j}Q_{i,j}\bigr|\le\sum_{i<j}|Q_{i,j}|\le(2^{o(\sqrt{n\log n})})^2\cdot 2^{-\Omega(n\log n)}\cdot 2^{2bn}\le 2^{-\Omega(n\log n)}\cdot 2^{2bn}$.

Define $S=\bigl(\{0,1\}^{bn}\times\{0,1\}^{bn}\bigr)\smallsetminus\bigcup_{i<j}Q_{i,j}$, and note that $S$ is a rectangle. Define the total function $F\colon S\to\{0,1\}$ by $F(x,y)=1$ iff $(x,y)\in R_i$ for some $i$; note that $F$ is consistent with $f\circ g^n$. Since the rectangles $R_i\cap S$ are pairwise disjoint, we have $\log\binr(F)\le o(\sqrt{n\log n})$ and thus the deterministic (in particular, conondeterministic) communication complexity of $F$ is $o(n\log n)$ by \eqref{nnr:eq:quadratic}.\footnote{Any improvement in the cost of converting unambiguous protocols to conondeterministic protocols (for total functions) would yield a corresponding quantitative improvement in \autoref{nnr:lem:lower}; however, the result of \cite{Goos} shows that the exponent of $2$ in such a conversion cannot be decreased below $1.128$.} Thus the set $W_{0^n}\cap S$ can be covered with $2^{o(n\log n)}$ many subrectangles of $S$ that are disjoint from all $R_i$'s. At least one of these subrectangles, call it $T$, covers at least a $2^{-o(n\log n)}$ fraction of $W_{0^n}\cap S$. Since $|W_{0^n}|\ge 2^{-n}\cdot 2^{2bn}$, we have $|W_{0^n}\cap S|\ge|W_{0^n}|-\bigl|\bigcup_{i<j}Q_{i,j}\bigr|\ge\frac{1}{2}\cdot|W_{0^n}|$. Thus $T$ also covers at least a $2^{-o(n\log n)}$ fraction of $W_{0^n}$. By \autoref{nnr:lem:conf}, there exists a subcube $Z\subseteq g^n(T)$ of codimension $o((n\log n)/\log n)=o(n)\le n/2$. But since any subcube of codimension $\le n/2$ contains a $1$-input of our $f$, this means $T$ contains a $1$-input of $f\circ g^n$, contradicting the fact that $T$ is disjoint from all $R_i$'s.
\end{proof}


\section{Total Matrices}
\label{nnr:sec:total}

In this section we prove \autoref{nnr:thm:main}. All matrices are tacitly assumed to be total in this section.


\subsection{A Lemma}
\label{nnr:sec:total:lemma}

We use the following lemma in the proof of \refmainone (but we remark that this is not the only property of nonnegative rank we use in the proof of \refmainone).

The \emph{support} of a matrix is the set of locations of nonzero entries. The support of a rank $1$ matrix is always a rectangle.

\begin{lemma}
\label{nnr:lem:leftover}
For every $0$-$1$ matrix $M$, every $\nnr$ decomposition $M=M^{(1)}+\cdots+M^{(r)}$ with supports $R_1,\ldots,R_r$, and every $i\in\{1,\ldots,r\}$, $R_i\smallsetminus\bigcup_{j\ne i}R_j$ is a rectangle.
\end{lemma}

\begin{proof}
Denote $R=R_i\smallsetminus\bigcup_{j\ne i}R_j$. Assume $(a,b),(c,d)\in R$ and $a\ne c$ and $b\ne d$. Then we have $M^{(i)}_{a,b}=M^{(i)}_{c,d}=1$, and since $M^{(i)}$ has rank $1$ we have $M^{(i)}_{a,d}\cdot M^{(i)}_{c,b}=M^{(i)}_{a,b}\cdot M^{(i)}_{c,d}=1$. Since every entry of $M^{(i)}$ is in the range $[0,1]$, this forces $M^{(i)}_{a,d}=M^{(i)}_{c,b}=1$ and so $(a,d),(c,b)\in R_i$. In turn, this forces $M^{(j)}_{a,d}=M^{(j)}_{c,b}=0$ for all $j\ne i$ and hence $(a,d),(c,b)\not\in R_j$. Thus $(a,d),(c,b)\in R$ and so $R$ is a rectangle.
\end{proof}

As an aside (not needed for \refmainone), here is a simple application illustrating \autoref{nnr:lem:leftover}.

\begin{observation}
\label{nnr:obs:bipartite}
Define the intersection graph of a collection of rectangles to have nodes representing the rectangles, and an edge between two nodes iff their rectangles intersect. Consider a $0$-$1$ matrix $M$ with $\nnr(M)=r$ and a $\nnr$ decomposition $M=M^{(1)}+\cdots+M^{(r)}$ with supports $R_1,\ldots,R_r$. If the intersection graph of the supports is bipartite, then $\nnr(M)=\binr(M)$.
\end{observation}

\begin{proof}
Consider a set of rectangles $\{R_i\}_{i\in S}$ that simultaneously forms an independent set and a vertex cover. For $i\not\in S$, define $R'_i=R_i\smallsetminus\bigcup_{j\in S}R_j$. We claim that the collection $\{R_i\}_{i\in S}\cup\{R'_i\}_{i\not\in S}$ forms a partition of the $1$'s of $M$ into $r$ rectangles. \emph{Coverage:} Since the collection of all $R_i$'s forms a cover of the $1$'s of $M$, it follows by definition that $\{R_i\}_{i\in S}\cup\{R'_i\}_{i\not\in S}$ also forms a cover. \emph{Disjointness:} The rectangles $R_i$ for $i\in S$ are disjoint from each other by the independent set property. The rectangles $R_i$ for $i\not\in S$ are disjoint from each other by the vertex cover property, so certainly the sets $R'_i$ for $i\not\in S$ are disjoint from each other. For $i\not\in S$ and $j\in S$, $R'_i$ and $R_j$ are disjoint by definition. \emph{Rectangles:} For $i\not\in S$, we have $R'_i=R_i\smallsetminus\bigcup_{j\ne i}R_j$ by the vertex cover property, and the latter set is a rectangle by \autoref{nnr:lem:leftover}.
\end{proof}


\subsection{Proof of \refmainone}
\label{nnr:sec:total:proof}

Consider the following two techniques for converting a $1$'s cover of $M$ by rectangles $R_1,\ldots,R_r$ into a $1$'s partition of $M$.\bigskip

\noindent\emph{Technique 1:} For some permutation $\pi$ on $\{1,\ldots,r\}$, use the sets $R_i\smallsetminus\bigcup_{j:\pi(j)<\pi(i)}R_j$ (for $i\in\{1,\ldots,r\}$).\bigskip

\noindent\emph{Technique 2:} Let $R_i=A_i\times B_i$. For $\alpha\in\{1,-1\}^r$, we say that $\bigcap_iA_i^{\alpha_i}$ is a \emph{type} of row (with respect to the cover $R_1,\ldots,R_r$), where $A_i^1=A_i$ and $A_i^{-1}=\overline{A_i}$. We say the type is \emph{nontrivial} if $\alpha$ is not all $-1$'s. Note that all rows of the same type are identical, and thus each nontrivial nonempty type has an associated rectangle that covers the $1$'s in all rows of that type. The technique is to use this collection of rectangles, over all the nontrivial nonempty types of rows, or do the analogous thing for nontrivial nonempty types of columns.\bigskip

Technique 1 is guaranteed to produce a small partition, but the sets are not guaranteed to be rectangles. Technique 2 is guaranteed to produce a partition into rectangles, but it is not guaranteed to be a small partition (it is small if many types of rows/columns are empty). Our approach to prove \refmainone is to argue that, for a cover arising as the supports of a nonnegative rank $3$ decomposition, \emph{either} Technique 1 works \emph{or} Technique 2 works.

In principle, \refmainone could be proved by brute force. The reason is because in a nonnegative rank $3$ decomposition, there can be at most seven nontrivial types of rows and seven nontrivial types of columns, and duplicate rows and columns can be deleted without changing the nonnegative rank or the binary rank. Hence if there were a counterexample to \refmainone, there would be a counterexample of size at most $7\times 7$. However, such a brute force argument would be unenlightening; our argument provides insight into why \refmainone is true.

\begin{definition}
\label{nnr:def:compatible}
We say a pair of rectangles $(R,Q)$ is \emph{compatible} if $Q\smallsetminus R$ is a rectangle, in other words, either $R\cap Q=\emptyset$, or the rows of $R$ are a superset of the rows of $Q$, or the columns of $R$ are a superset of the columns of $Q$.
\end{definition}

\begin{proof}[Proof of \refmainone]
Of course, $\binr(M)\ge\nnr(M)$ for all $0$-$1$ matrices $M$, so we just need to show that if $\nnr(M)\le 3$ then $\binr(M)\le\nnr(M)$. The case $\nnr(M)\le 1$ is trivial. Suppose $\nnr(M)=2$ and let $M=M^{(1)}+M^{(2)}$ be a $\nnr$ decomposition with supports $R_1,R_2$. By \autoref{nnr:lem:leftover}, $R_2\smallsetminus R_1$ is a rectangle, and hence $R_1$ and $R_2\smallsetminus R_1$ are two rectangles forming a $1$'s partition of $M$, so $\binr(M)\le 2$.

Suppose $\nnr(M)=3$ and let $M=M^{(1)}+M^{(2)}+M^{(3)}$ be a $\nnr$ decomposition with supports $R_1,R_2,R_3$. First assume two of these rectangles form a compatible pair, say $(R_1,R_2)$. Then $R_2\smallsetminus R_1$ is a rectangle, and $R_3\smallsetminus(R_1\cup R_2)$ is a rectangle by \autoref{nnr:lem:leftover}. Hence $R_1$ and $R_2\smallsetminus R_1$ and $R_3\smallsetminus(R_1\cup R_2)$ are three rectangles forming a $1$'s partition of $M$, so $\binr(M)\le 3$.

Now assume no two of the rectangles $R_1,R_2,R_3$ form a compatible pair. We show that either there are only three nontrivial nonempty types of rows, or there are only three nontrivial nonempty types of columns (which gives a $1$'s partition into three rectangles by Technique 2). Say $R_i=A_i\times B_i$. Since neither $(R_1,R_2)$ nor $(R_2,R_1)$ are compatible, the following sets are all nonempty: $A_1\smallsetminus A_2$, $A_1\cap A_2$, $A_2\smallsetminus A_1$, $B_1\smallsetminus B_2$, $B_1\cap B_2$, $B_2\smallsetminus B_1$. Now observe that either $(A_1\cap A_2)\times(B_1\smallsetminus B_2)\subseteq R_3$ or $(A_1\smallsetminus A_2)\times(B_1\cap B_2)\subseteq R_3$, since if not then letting $(a,b)\in\bigl((A_1\cap A_2)\times(B_1\smallsetminus B_2)\bigr)\smallsetminus R_3$ and $(c,d)\in\bigl((A_1\smallsetminus A_2)\times(B_1\cap B_2)\bigr)\smallsetminus R_3$, we have $(a,b),(c,d)\in R_1\smallsetminus(R_2\cup R_3)$ but $(a,d)\in R_1\cap R_2$ and hence $(a,d)\not\in R_1\smallsetminus(R_2\cup R_3)$, so $R_1\smallsetminus(R_2\cup R_3)$ is not a rectangle, contradicting \autoref{nnr:lem:leftover}. Similarly, either $(A_1\cap A_2)\times(B_2\smallsetminus B_1)\subseteq R_3$ or $(A_2\smallsetminus A_1)\times(B_1\cap B_2)\subseteq R_3$.

Henceforth assume that $(A_1\cap A_2)\times(B_1\smallsetminus B_2)\subseteq R_3$ (if $(A_1\smallsetminus A_2)\times(B_1\cap B_2)\subseteq R_3$ then a symmetric argument applies). If $(A_2\smallsetminus A_1)\times(B_1\cap B_2)\subseteq R_3$ then $B_1\subseteq B_3$ and thus $(R_3,R_1)$ would be compatible, so we may henceforth assume that $(A_1\cap A_2)\times(B_2\smallsetminus B_1)\subseteq R_3$.

Consider three cases: $(A_1\smallsetminus A_2)\cap A_3=\emptyset$ or $A_1\smallsetminus A_2\subseteq A_3$ or neither. In the second case, $A_1\subseteq A_3$ and thus $(R_3,R_1)$ would be compatible. In the third case, we claim that we must have $B_1\subseteq B_3$, and thus $(R_3,R_1)$ would be compatible. To prove the claim, first note that $R_1\smallsetminus(R_2\cup R_3)=\bigl((A_1\smallsetminus A_2)\times B_1\bigr)\smallsetminus R_3$ since $(A_1\cap A_2)\times(B_1\smallsetminus B_2)\subseteq R_3$. We have $\emptyset\subsetneq(A_1\smallsetminus A_2)\cap A_3\subsetneq A_1\smallsetminus A_2$ (since we are assuming the third case) and, if we assume $B_1\not\subseteq B_3$, we have $\emptyset\subsetneq B_1\cap B_3\subsetneq B_1$ (since $\emptyset\subsetneq B_1\smallsetminus B_2\subseteq B_3$). Together, these imply that $\bigl((A_1\smallsetminus A_2)\times B_1\bigr)\smallsetminus R_3$ is not a rectangle, contradicting \autoref{nnr:lem:leftover}. So, we may henceforth assume the first case, namely $(A_1\smallsetminus A_2)\cap A_3=\emptyset$. Similarly, we may henceforth assume $(A_2\smallsetminus A_1)\cap A_3=\emptyset$.

To finish the proof, we consider two cases:
\begin{itemize}
\item[(i)] $B_1\cap B_2\cap B_3\ne\emptyset$, or
\item[(ii)] $B_1\cap B_2\cap B_3=\emptyset$.
\end{itemize}
First assume (i) holds. We claim that $B_1\cap B_2\subseteq B_3$ which, together with $B_1\smallsetminus B_2\subseteq B_3$, implies that $B_1\subseteq B_3$ and thus $(R_3,R_1)$ would be compatible (in fact, $(R_3,R_2)$ would also be compatible). Let $b\in B_1\cap B_2\cap B_3$ and suppose for contradiction that some $d\in(B_1\cap B_2)\smallsetminus B_3$. Let $a\in A_1\smallsetminus A_2$ and $c\in A_1\cap A_2$ and $e\in A_2\smallsetminus A_1$. Since $(A_1\smallsetminus A_2)\cap A_3=\emptyset$ we have $(a,b),(a,d)\in R_1\smallsetminus(R_2\cup R_3)$ and thus $M^{(1)}_{a,b}=M^{(1)}_{a,d}=1$. Since $M^{(1)}$ has rank $1$, this implies that $M^{(1)}_{c,b}=M^{(1)}_{c,d}$. Similarly, $M^{(2)}_{c,b}=M^{(2)}_{c,d}$ (by using $e$ in place of $a$). However, $(c,b)\in R_3$ (since $A_1\cap A_2\subseteq A_3$) and $(c,d)\not\in R_3$ (since $d\not\in B_3$), and thus $M^{(3)}_{c,b}>0=M^{(3)}_{c,d}$. Hence \[M_{c,b}~=~M^{(1)}_{c,b}+M^{(2)}_{c,b}+M^{(3)}_{c,b}~>~M^{(1)}_{c,d}+M^{(2)}_{c,d}+M^{(3)}_{c,d}~=~M_{c,d}\] which is a contradiction since $M_{c,b}=1=M_{c,d}$.

Now assume (ii) holds. In this case, $(R_1\cup R_2)\cap R_3=(A_1\cap A_2)\times\bigl((B_1\smallsetminus B_2)\cup(B_2\smallsetminus B_1)\bigr)$. If $A_3\smallsetminus(A_1\cup A_2)\ne\emptyset$ and $B_3\smallsetminus(B_1\cup B_2)\ne\emptyset$ then $R_3\smallsetminus(R_1\cup R_2)$ is not a rectangle (since if $a\in A_3\smallsetminus(A_1\cup A_2)$, $b\in B_3\smallsetminus(B_1\cup B_2)$, $c\in A_1\cap A_2$, and $d\in B_1\smallsetminus B_2$, then $(a,d),(c,b)\in R_3\smallsetminus(R_1\cup R_2)$ but $(c,d)\in R_1$), contradicting \autoref{nnr:lem:leftover}. If $A_3\smallsetminus(A_1\cup A_2)=\emptyset$ then $A_3=A_1\cap A_2$ and thus there are only three nontrivial nonempty types of rows ($A_1\cap\overline{A_2}\cap\overline{A_3}$ and $A_1\cap A_2\cap A_3$ and $\overline{A_1}\cap A_2\cap\overline{A_3}$). On the other hand, if $B_3\smallsetminus(B_1\cup B_2)=\emptyset$ then $B_3=(B_1\smallsetminus B_2)\cup(B_2\smallsetminus B_1)$ and thus there are only three nontrivial nonempty types of columns ($B_1\cap\overline{B_2}\cap B_3$ and $B_1\cap B_2\cap\overline{B_3}$ and $\overline{B_1}\cap B_2\cap B_3$).
\end{proof}


\subsection{Some Examples}
\label{nnr:sec:total:examples}

We now give some examples that elucidate features of the proof of \refmainone.

\[\left[\begin{array}{cccc}
1&1&0&1\\
1&0&1&1\\
0&1&1&1\\
1&1&1&1
\end{array}\right]~~~~~~~~~~~~
\left[\begin{array}{cccc}
0&1&0&0\\
1&1&1&0\\
0&1&1&1\\
0&0&1&1
\end{array}\right]~~~~~~~~~~~~
\left[\begin{array}{ccc}
1&1&0\\
1&0&1\\
0&1&1\\
1&1&1
\end{array}\right]\]

The first matrix above shows that \autoref{nnr:lem:leftover} cannot possibly be the only property of $\nnr$ decompositions used to prove \refmainone, because this matrix has a $1$'s cover of three rectangles satisfying the conclusion of \autoref{nnr:lem:leftover} ($\{1,4\}\times\{1,2,4\}$ and $\{2,4\}\times\{1,3,4\}$ and $\{3,4\}\times\{2,3,4\}$), yet the binary rank is $4$.

The second matrix above shows that Technique 2 is not sufficient to prove \refmainone, because the binary rank is $3$, yet there are four distinct nonzero rows and four distinct nonzero columns (hence necessarily four nontrivial nonempty types of rows and four nontrivial nonempty types of columns in any $1$'s cover by rectangles).

The third matrix above shows that Technique 1 is not sufficient for converting an arbitrary nonnegative rank $3$ decomposition into a $1$'s partition with the same number of rectangles, since this matrix has a nonnegative rank decomposition with three matrices, no two of whose supports are compatible.


\subsection{Separations}
\label{nnr:sec:total:gap}

\begin{proof}[Proof of \refmaintwo]
Consider the following $5\times 6$ matrix $M$ consisting of all possible columns having a $1$ in the bottom row and two $1$'s among the top four rows. (In fact, any column of $M$ can be deleted to yield a $5\times 5$ matrix that still works, but it seems cleaner to include all the columns.)
\definecolor{mygreen}{RGB}{0,191,0}
\[M~=~\left[\begin{array}{cccccc}
{\color{red}1}&{\color{mygreen}1}&{\color{blue}1}&0&0&0\\
{\color{mygreen}1}&0&0&{\color{red}1}&{\color{blue}1}&0\\
0&{\color{red}1}&0&{\color{blue}1}&0&{\color{mygreen}1}\\
0&0&{\color{mygreen}1}&0&{\color{red}1}&{\color{blue}1}\\
1&1&1&1&1&1
\end{array}\right]\hspace{2.5cm}
U~=~\left[\begin{array}{cccc}
1&0&0&0\\
0&1&0&0\\
0&0&1&0\\
0&0&0&1\\
\frac{1}{2}&\frac{1}{2}&\frac{1}{2}&\frac{1}{2}
\end{array}\right]\]
We have $\nnr(M)\ge 4$ since $M$ has a fooling set of size $4$. We also have $\nnr(M)\le 4$ since $M=UV$ where $U$ is as above and $V$ is the top four rows of $M$. Clearly $\binr(M)\le 5$, so we just need to argue that $\binr(M)\ge 5$. There are several (ad hoc) ways to see this; our preferred way is as follows. Note that the $1$'s of the top four rows can be partitioned into three fooling sets of size four, indicated by the colors red, green, and blue.\footnote{In case the colors are not visible: red is entries $(1,1),(2,4),(3,2),(4,5)$, green is entries $(1,2),(2,1),(3,6),(4,3)$, and blue is entries $(1,3),(2,5),(3,4),(4,6)$, where row $1$ is topmost and column $1$ is leftmost.} Suppose for contradiction there is a partition of the $1$'s of $M$ into just four rectangles. Then each of those rectangles must contain exactly one red $1$, one green $1$, and one blue $1$, and hence must contain all three $1$'s in one of the top four rows (since these are the only size-three rectangles within the top four rows). Hence each pair of rectangles in the partition shares one column in common. Thus only one of the four rectangles can touch the bottom row; but since this rectangle is three columns wide, that leaves three $1$'s in the bottom row uncovered, which is a contradiction.
\end{proof}

The upper bound $\nnr(M)\le 4$ for \refmaintwo can be phrased in terms of a sampling protocol as follows. One party (it does not matter which) picks one of the top four rows uniformly at random and sends this index to the other party (so $2$ bits of communication, $4$ possible transcripts). Then independently of each other: Alice outputs the chosen row with probability $2/3$ and the bottom row with probability $1/3$. Bob outputs uniformly at random one of the three columns in which the chosen row has a $1$ in $M$.

\begin{proof}[Proof of \refmainthree]
Consider the following $12\times 12$ matrix $M$.
\definecolor{mygreen}{RGB}{0,191,0}
\[M~=~\left[\begin{array}{cccccccccccc}
{\color{red}1}&{\color{red}1}&{\color{red}0}&0&0&0&0&0&0&1&0&0\\
{\color{red}1}&{\color{red}0}&{\color{red}1}&0&0&0&0&0&0&0&1&0\\
{\color{red}0}&{\color{red}1}&{\color{red}1}&0&0&0&0&0&0&0&0&1\\
0&0&0&{\color{mygreen}1}&{\color{mygreen}1}&{\color{mygreen}0}&0&0&0&1&0&0\\
0&0&0&{\color{mygreen}1}&{\color{mygreen}0}&{\color{mygreen}1}&0&0&0&0&1&0\\
0&0&0&{\color{mygreen}0}&{\color{mygreen}1}&{\color{mygreen}1}&0&0&0&0&0&1\\
0&0&0&0&0&0&{\color{blue}1}&{\color{blue}1}&{\color{blue}0}&1&0&0\\
0&0&0&0&0&0&{\color{blue}1}&{\color{blue}0}&{\color{blue}1}&0&1&0\\
0&0&0&0&0&0&{\color{blue}0}&{\color{blue}1}&{\color{blue}1}&0&0&1\\
{\color{red}1}&{\color{red}1}&{\color{red}1}&1&1&1&0&0&0&1&1&1\\
0&0&0&{\color{mygreen}1}&{\color{mygreen}1}&{\color{mygreen}1}&1&1&1&1&1&1\\
1&1&1&0&0&0&{\color{blue}1}&{\color{blue}1}&{\color{blue}1}&1&1&1
\end{array}\right]\hspace{1cm}
U~=~\left[\begin{array}{ccccccccc}
1&0&0&0&0&0&0&0&0\\
0&1&0&0&0&0&0&0&0\\
0&0&1&0&0&0&0&0&0\\
0&0&0&1&0&0&0&0&0\\
0&0&0&0&1&0&0&0&0\\
0&0&0&0&0&1&0&0&0\\
0&0&0&0&0&0&1&0&0\\
0&0&0&0&0&0&0&1&0\\
0&0&0&0&0&0&0&0&1\\
\frac{1}{2}&\frac{1}{2}&\frac{1}{2}&\frac{1}{2}&\frac{1}{2}&\frac{1}{2}&0&0&0\\
0&0&0&\frac{1}{2}&\frac{1}{2}&\frac{1}{2}&\frac{1}{2}&\frac{1}{2}&\frac{1}{2}\\
\frac{1}{2}&\frac{1}{2}&\frac{1}{2}&0&0&0&\frac{1}{2}&\frac{1}{2}&\frac{1}{2}
\end{array}\right]\]
We have $\nnr(M)\ge 9$ since $M$ has a fooling set of size $9$. We also have $\nnr(M)\le 9$ since $M=UV$ where $U$ is as above and $V$ is the top nine rows of $M$. Clearly $\binr(M)\le 12$, so we just need to argue that $\binr(M)\ge 12$. Suppose for contradiction there is a partition of the $1$'s of $M$ into just $11$ rectangles $R_1,\ldots,R_{11}$. Let $Q_r=\{1,2,3,10\}\times\{1,2,3\}$ be the red rectangle of $M$, let $Q_g=\{4,5,6,11\}\times\{4,5,6\}$ be the green rectangle of $M$, and let $Q_b=\{7,8,9,12\}\times\{7,8,9\}$ be the blue rectangle of $M$. Note that these three rectangles form a fooling set. (We say $A\times B$ and $C\times D$ are fooling if either $A\times D$ or $C\times B$ is an all $0$'s submatrix.) Hence no $R_i$ can intersect more than one of $Q_r,Q_g,Q_b$. By the pigeonhole principle, at least one of $Q_r,Q_g,Q_b$, say $Q_r$, intersects at most three of the $R_i$'s. By inspection, the only way to partition the $1$'s of $Q_r$ using at most three rectangles is to use $\{1,2,10\}\times\{1\}$, $\{1,3,10\}\times\{2\}$, and $\{2,3,10\}\times\{3\}$. Hence three of the $R_i$'s, say $R_1,R_2,R_3$, intersect $Q_r$ in those three subrectangles. Thus since each of $R_1,R_2,R_3$ touches two of the first three rows of $M$, none of $R_1,R_2,R_3$ can touch any of the last three columns of $M$ (otherwise a $0$-entry would be covered by an $R_i$). Thus the $1$-entries $M_{1,10}$, $M_{2,11}$, $M_{3,12}$ must be covered by some other distinct $R_i$'s, say $R_4,R_5,R_6$, none of which can intersect $Q_g$ or $Q_b$ since $\{1,2,3\}\times\{4,5,6,7,8,9\}$ is all $0$'s. This leaves $R_7,\ldots,R_{11}$ to cover the $1$'s of $Q_g$ and $Q_b$, which is a contradiction since three rectangles are needed for $Q_g$, and another three are needed for $Q_b$.
\end{proof}

The upper bound $\nnr(M)\le 9$ for \refmainthree can be phrased in terms of a sampling protocol as follows. One party (it does not matter which) picks $i\in\{1,\ldots,9\}$ uniformly at random and sends $i$ to the other party. Then independently of each other: Alice outputs row $i$ with probability $1/2$, and with probability $1/4$ each outputs one of the two rows in which column $i$ has a $\frac{1}{2}$ in $U$. Bob outputs uniformly at random one of the three columns in which row $i$ has a $1$ in $M$.


\section{Open Questions}
\label{nnr:sec:open}

Can \autoref{nnr:thm:partialsep} or \autoref{nnr:lem:lower} be quantitatively improved? Can it be witnessed by a composed function with a simpler gadget than the confounding gadget (and with an elementary proof that avoids the machinery of \cite{GLMWZ,Goos})?

What is the best separation between log of nonnegative rank and log of binary rank for total $0$-$1$ matrices? We know it is at least an additive constant and at most quadratic.

A number of so-called simulation theorems are known, which convert lower bounds for query complexity measures into lower bounds for the corresponding communication complexity measures (e.g., \cite{She,GLMWZ,GPW}). Unfortunately, no such simulation theorem is known for unambiguous nondeterminism; such a result could be useful for obtaining new binary rank lower bounds.

For total $0$-$1$ matrices $M,N$, let $M\otimes N$ denote the (Kronecker) tensor product, and let $M^{\otimes k}$ denote the $k$-fold tensor product of $M$ with itself. How do $\nnr$ and $\binr$ behave under the tensor product? This can be viewed as a direct sum question. We trivially have \[\max\bigl\{\nnr(M)\cdot\fool(N),~\fool(M)\cdot\nnr(N)\bigr\}~\le~\nnr(M\otimes N)~\le~\nnr(M)\cdot\nnr(N)\] and  \[\max\bigl\{\binr(M)\cdot\fool(N),~\fool(M)\cdot\binr(N)\bigr\}~\le~\binr(M\otimes N)~\le~\binr(M)\cdot\binr(N)\] where $\fool(M)$ denotes the largest size of a fooling set of $M$. It does not seem to be known whether there exists an $M$ and $k$ for which $\binr(M^{\otimes k})<\binr(M)^k$. It is known \cite{BKLT} that there exist total real nonnegative matrices $M,N$ such that $\nnr(M\otimes N)<\nnr(M)\cdot\nnr(N)$.

\begin{conjecture}
\label{nnr:conj:tensor}
For all $k$, $\binr(M^{\otimes k})=5^k$ where matrix $M$ is from the proof of \refmaintwo.
\end{conjecture}

For this particular $M$, we have $\nnr(M^{\otimes k})=\fool(M^{\otimes k})=4^k$ for all $k\ge 1$. Hence under this conjecture, we would have a family of total matrices for which unambiguous nondeterminism requires $\log_4(5)\approx 1.16$ times more bits of communication than the corresponding sampling problem. We have been unable to make any significant progress on the conjecture.

Is there some other way (besides possibly tensor product) to amplify the gap and get a family of total $0$-$1$ matrices exhibiting some separation between nonnegative rank and binary rank? It is conceivable that a computer search could be used to find a better gap example (e.g., by confirming \autoref{nnr:conj:tensor} for $k=2$). However, we feel that such an enterprise would be unenlightening and unlikely to lead to any proof techniques for obtaining a superlinear gap between nonnegative rank and binary rank.

A measure sandwiched between $\nnr$ and $\binr$ is the number of rectangles needed to uniformly cover the $1$'s of a matrix (so all $1$'s are covered by the same number of rectangles). How does this measure relate to $\nnr$ and $\binr$? The upper bound in \autoref{nnr:thm:partialsep} also works for this intermediate measure, but the upper bounds in \refmaintwo and \refmainthree do not.


\subsection*{Acknowledgments}

I thank Mika G\"o\"os, Troy Lee, and Toniann Pitassi for discussions and anonymous reviewers for comments.


\bibliographystyle{alpha}
\bibliography{nnr}

\end{document}